\documentclass[journal]{IEEEtran}

\usepackage{amsmath,amsfonts,amssymb,amsthm}
\usepackage{graphicx}
\usepackage{cite}
\usepackage{bm}
\usepackage{dsfont}
\usepackage[caption=false,font=footnotesize]{subfig}
\usepackage{siunitx}

\theoremstyle{plain}
\newtheorem{proposition}{Proposition}
\newtheorem{theorem}{Theorem}
\newtheorem{corollary}{Corollary}
\newtheorem{lemma}{Lemma}

\theoremstyle{definition}
\newtheorem{definition}{Definition}

\theoremstyle{remark}
\newtheorem{remark}{Remark}

\newcommand{\euler}{\mathrm{e}}
\renewcommand{\d}{\mathrm{d}}
\newcommand{\given}{\,|\,}
\DeclareMathOperator{\N}{N}
\DeclareMathOperator{\U}{U}
\DeclareMathOperator{\D}{D}
\DeclareMathOperator{\erf}{erf}

\def\X{{\mathcal X}}

\def\R{{\mathbb R}}

\def\E{{\mathbb E}}
\def\V{{\mathbb V}}
\def\P{{\mathbb P}}

\begin{document}

\title{Level Crossing Rate Bounds for Products of Independent Nonstationary Fading Processes}

\author{
    Plínio Santini Dester%
    \thanks{P. S. Dester (plinio.dester@unesp.br) is with the São Paulo State University (UNESP), School of Engineering, São João da Boa Vista, Brazil.}
}

\markboth{}
{Dester: Level Crossing Rate Bounds for Products of Independent
Nonstationary Fading Processes}

\maketitle

\begin{abstract}
    We develop a unified framework for the instantaneous total level crossing rate (LCR) of products of independent positive nonstationary fading processes. Universal upper and lower bounds are first derived from an exact decomposition of the product derivative and signed crossing imbalance, requiring only factor-wise crossing statistics.
    We then introduce the drifting Gaussian-derivative (DGD) class, in which each conditional derivative is Gaussian with possibly nonzero drift. For products of DGD factors, an exact Gaussian-mixture representation yields bounds that jointly capture aggregate drift and derivative fluctuations. The lower bound reduces under stationarity to a known \(\ell_2\) bound, whereas, to the best of the author's knowledge, the upper bound is new even in the stationary case.
    Furthermore, the bounds are asymptotically exact in drift-dominated regimes.
    The framework is specialized to nonstationary Double-Rayleigh fading and to quasi-identically distributed DGD factors, with closed forms for products of \(\alpha\)-\(\mu\) envelopes. Numerical results show that the DGD bounds closely bracket the exact LCR.
\end{abstract}

\begin{IEEEkeywords}
    Level crossing rate, nonstationary fading, product processes, cascaded
    fading channels, drifting Gaussian-derivative processes.
\end{IEEEkeywords}

\section{Introduction}
\label{sec:introduction}

First-order fading statistics, such as the probability density function
(PDF) and the outage probability, quantify the severity of signal
attenuation but do not describe how fading events evolve over time.
The level crossing rate (LCR), defined as the expected number of crossings
of a prescribed level per unit time, provides this complementary temporal
information. Together with the average fade duration, the LCR is relevant
to the characterization of error bursts, the design of interleavers and
coding schemes, and the evaluation of adaptive transmission mechanisms
\cite{rice1944mathematical,patzold1999level}.
The general analytical foundation is Rice's formula, which expresses the
crossing intensity in terms of the joint distribution of a process and its
time derivative \cite{rice1944mathematical}. Classical results, applications,
and analytical techniques for level-crossing problems are surveyed in
\cite{blake1973level}, while a rigorous modern treatment of Rice formulas
and level sets of random processes is provided in \cite{azais2009level}.

Multiplicative fading envelopes arise naturally when propagation involves
several successive or superimposed attenuation mechanisms. Important
examples include cascaded multihop channels and MIMO keyhole channels. The product structure has motivated
general cascaded fading models, such as the \(N*\)Nakagami model
\cite{karagiannidis2007cascaded}, as well as specific investigations of the
LCR and average fade duration of \mbox{Double-Nakagami-\(m\)} and multihop Rayleigh
processes \cite{zlatanov2008double,hadzivelkov2009multihop}. These studies
demonstrate the relevance of second-order statistics for product channels,
but also illustrate their analytical difficulty. Even when the constituent
processes have simple fading distributions, the product LCR generally
involves multidimensional integrals or distribution-specific
approximations.

Most available product-LCR formulations are developed under stationarity.
This assumption substantially simplifies the analysis because the upward
and downward crossing rates coincide and the conditional derivative of each
fading envelope has zero mean. However, stationarity may be restrictive in
high-mobility, large-array, and rapidly changing propagation environments,
for which channel measurements and models exhibit temporal, spatial, or
frequency nonstationarity
\cite{wang2018survey,cheng2022nonstationarity}. In a nonstationary setting,
the crossing rate must be characterized instantaneously. Moreover, upward
and downward crossings need not be balanced and their difference represents
the instantaneous probability flux through the considered level \cite{azais2009level}. For a
product process, the drifts of the constituent factors may reinforce or
cancel one another, making a direct extension of stationary results
nontrivial.

The instantaneous LCR of particular nonstationary fading processes has
received comparatively limited attention. P{\"a}tzold \emph{et al.}
considered a nonstationary Rayleigh envelope defined as the magnitude of a
non-wide-sense-stationary complex Gaussian process with time-varying
second-order statistics \cite{patzold2017nonstationary}. By applying the
time-dependent Rice formula, they obtained a closed-form expression for the
instantaneous LCR and specialized it to a sum-of-chirps multipath model
capturing variations in the mobile velocity. More broadly, physically
motivated nonstationary fading models have been developed using
time-varying geometries and random mobile trajectories
\cite{borhani2017trajectory}. These studies establish the relevance of
instantaneous crossing statistics in nonstationary propagation
environments, but remain tied to particular channel constructions and
single fading envelopes. In contrast, a general treatment of the LCR of
products of independent nonstationary processes has not been available.

Recently, sharp LCR inequalities were obtained for products and ratios of independent
stationary processes, together with an exact formulation when the
constituent conditional derivatives are Gaussian
\cite{dester2025level, dester2026fundamental}. The stationary theory relates the product LCR to
averaged contributions of the individual factors and leads to concise
\(\ell_\infty\)-, \(\ell_2\)-, and \(\ell_1\)-norm bounds. Nevertheless,
stationarity eliminates the drift terms that govern the crossing imbalance.
Consequently, new arguments are required to characterize products of
independent nonstationary processes and to quantify the combined effects
of deterministic drift and random derivative fluctuations.

In many classical fading models, the envelope derivative conditioned on the
current envelope is Gaussian. Their usual stationary formulations correspond
to zero conditional drift, whereas time-varying deterministic scaling may
introduce a nonzero conditional drift while preserving the conditional
Gaussian structure. 
This motivates the drifting Gaussian-derivative (DGD) class considered in
this paper. More broadly, the class encompasses many classical fading and
shadowing processes generated from differentiable Gaussian processes,
including, under their standard temporal constructions, Rayleigh, Rice,
Nakagami-\(m\), Weibull, \(\alpha\)-\(\mu\), \(\kappa\)-\(\mu\), and
lognormal models \cite{yacoub2007alpha,yacoub2007kappa,patzold2010shadow}.

In this paper, we develop a general framework for the instantaneous total
LCR of a product of independent positive nonstationary processes. The main
contributions are summarized as follows.

\begin{enumerate}
    \item We derive universal upper and lower bounds for the
    total LCR of a product process. The signed crossing imbalance is shown
    to decompose exactly into the sum of averaged factor contributions.
    The upper bound is the sum of the averaged individual LCRs, whereas the
    lower bound contains an explicit drift-correction term that captures
    reinforcement or cancellation among the remaining factors. In the
    stationary specialization, these results recover the
    \(\ell_\infty\)- and \(\ell_1\)-norm bounds of
    \cite{dester2025level}.

    \item We introduce the DGD class, in which the derivative of each envelope, conditioned on its current value, is Gaussian with a possibly nonzero state-dependent mean.
    Under this class, we derive a lower bound that
    combines the aggregate conditional drift with the Euclidean (or $\ell_2$) norm of the
    averaged derivative spreads. In the stationary case, the drift vanishes
    and the result reduces to the \(\ell_2\) lower bound of
    \cite{dester2025level}. We also derive a tractable upper bound from the
    folded-normal inequality
    \[
        \E\!\left[|a+bZ|\right]
        \le
        \sqrt{a^2+\frac{2}{\pi}b^2},
        \qquad Z\sim\mathcal N(0,1).
    \]
    To the best of the author's knowledge, the resulting upper bound is new even in the stationary case.

    \item We specialize the framework to a nonstationary
    Double-Rayleigh channel with time-varying deterministic gains. We
    obtain an exact one-dimensional integral, closed-form expressions for
    the factor-wise contributions, explicit universal and DGD bounds, and
    an asymptotic characterization of their relative gap. We then consider
    quasi-identically distributed products of DGD factors with linear conditional
    drifts and a common derivative-spread shape. In this setting, the
    deterministic drifts combine through their algebraic sum, whereas the
    independent Gaussian derivative spreads combine through their
    Euclidean norm. Closed-form expressions are obtained for the product of $n$
    \(\alpha\)-\(\mu\) envelopes.
\end{enumerate}

For the considered Double-Rayleigh examples, the numerical results show
that the DGD bounds substantially tighten the universal bounds.
Moreover, the DGD bounds become asymptotically tight when the aggregate
drift dominates the derivative fluctuations. On the other hand, the universal bounds become asymptotically exact when there is a dominant LCR component.

The remainder of the paper is organized as follows.
Section~\ref{sec:prelim} introduces the instantaneous upward, downward, and
total crossing rates and the signed crossing imbalance.
Section~\ref{sec:universal_bounds} derives the universal nonstationary
product bounds.
Section~\ref{sec:DGD_Bound} presents the DGD exact formulation and bounds.
Section~\ref{sec:application} applies the results to the nonstationary
Double-Rayleigh model and to quasi-identically distributed products of DGD factors, for which we obtain closed forms for the product of $n$ $\alpha$-$\mu$ fading models.
Finally, Section~\ref{sec:conclusion} concludes the paper.

\textit{Notation:}
The symbol \(\triangleq\) denotes equality by definition.
For \(x\in\R\), \(x^+\triangleq\max\{x,0\}\),
\(x^-\triangleq\max\{-x,0\}\), and
\(\operatorname{sgn}(x)\) denotes the sign function, with
\(\operatorname{sgn}(0)=0\).
The symbols \(\P\), \(\E\), and \(\V\) denote probability measure, expectation operator, and variance operator, respectively.
For random variables $X$ and $Y$,
the PDFs of \(X\) and \((X,Y)\) are denoted by \(f_X\) and \(f_{X,Y}\),
respectively, and \(X\given Y\) denotes \(X\) conditioned on \(Y\).
We write \(X\sim\mathcal N(\mu,\sigma^2)\) to indicate that \(X\)
is normally distributed with mean \(\mu\) and variance \(\sigma^2\).
For real-valued functions \(f\) and \(g\), the notation
\({f(x)=\mathcal O(g(x))}\) as \(x\to a\) means that
\(\lvert f(x)/g(x)\rvert\) remains bounded in a neighborhood of the
relevant limit point \(a\).
The symbol \(\|\cdot\|_p\) denotes the standard \(\ell_p\) norm.
Calligraphic letters, such as \(\mathcal X\), denote stochastic processes,
whereas the corresponding capital italic letters, such as \(X(t)\), denote
their values at a fixed time, and a dot, as in \(\dot X(t)\), denotes differentiation with respect to time.
The math boldface font is used to denote vectors, e.g., a vector of real numbers $\mathbf{x} = (x_1,\dots,x_n)$ or a vector of random variables $\mathbf{X} = (X_1,\dots,X_n)$.

\section{Preliminaries}
\label{sec:prelim}

Throughout, every individual or product process entering an LCR expression is assumed to satisfy the following conditions.

\begin{definition}\label{def:rice-regular}
    Let ${\mathcal{X} = \{X(t): t\in T\}}$ be a positive real-valued continuous-time random process, where $T$ is an interval of $\R$. We say $\mathcal X$ is \textit{Rice-regular} when for every $t$ in the interior of $T$, $\mathcal{X}$ satisfies the following properties:
    \begin{enumerate}
        \item $\mathcal{X}$ has a.s. continuously differentiable sample paths;
        \item the densities of $X(t)$ and $\dot X(t)$ and the joint density of $(X(t),\dot X(t))$ exist with respect to the Lebesgue measure and are continuous with respect to $t$;
        \item for every crossing level $u>0$,
        \[
            \int_{\R} |\dot x| f_{X(t),\dot X(t)}(u,\dot x)\,\d \dot x < \infty .
        \]
    \end{enumerate}
\end{definition}

These assumptions contemplate practical fading scenarios.
However, for less restrictive assumptions regarding the existence of the expected number of level crossings, see Aza{\"\i}s and Wschebor~\cite[Section~3.2.3]{azais2009level}.

For Rice-regular processes, the upward, downward, and total instantaneous level crossing rates of $\mathcal X$ at level $u$ and time $t$ exist and are given by \cite[Section~3.2]{azais2009level}
\begin{align}
    \U_\X(u;t)
    &=
    \int_{0}^{\infty}
    \dot x\, f_{X(t),\dot X(t)}(u,\dot x)\,\d \dot x, \label{eq:U_def_int}\\
    \D_\X(u;t)
    &=
    \int_{-\infty}^{0}
    (-\dot x)\, f_{X(t),\dot X(t)}(u,\dot x)\,\d \dot x, \label{eq:D_def_int}\\
    \N_\X(u;t)
    &=
    \int_{-\infty}^{\infty}
    |\dot x|\, f_{X(t),\dot X(t)}(u,\dot x)\,\d \dot x. \label{eq:N_def_int}
\end{align}
Equivalently, for nontrivial levels with $f_{X(t)}(u)>0$,
\begin{align}
    \U_\X(u;t)
    &=
    \E\!\left[(\dot X(t))^+ \given X(t)=u\right] f_{X(t)}(u), \label{eq:U_def}\\
    \D_\X(u;t)
    &=
    \E\!\left[(\dot X(t))^- \given X(t)=u\right] f_{X(t)}(u), \label{eq:D_def}\\
    \N_\X(u;t)
    &=
    \E\!\left[|\dot X(t)| \given X(t)=u\right] f_{X(t)}(u), \label{eq:N_def}
\end{align}
Hence,
\begin{equation}
    \N_\X(u;t)=\U_\X(u;t)+\D_\X(u;t).
\end{equation}

We also define the signed crossing imbalance
\begin{align}
\Delta_{\mathcal X}(u;t)
&\triangleq
\U_{\mathcal X}(u;t)-\D_{\mathcal X}(u;t) \label{eq:delta_def}\\
&=
\int_{-\infty}^\infty
\dot x f_{X(t),\dot X(t)}(u,\dot x)\,\d \dot x \nonumber\\
&=
\E\!\left[\dot X(t)\given X(t)=u\right]f_{X(t)}(u).
\label{eq:delta_X}
\end{align}
Consequently,
\begin{align}
\U_{\mathcal X}(u;t)
=
\frac{\N_{\mathcal X}(u;t)+\Delta_{\mathcal X}(u;t)}{2}, \label{eq:U_from_NDelta}\\
\D_{\mathcal X}(u;t)
=
\frac{\N_{\mathcal X}(u;t)-\Delta_{\mathcal X}(u;t)}{2}.
\label{eq:D_from_NDelta}
\end{align}
Therefore, it is sufficient to characterize the total crossing rate
$\N_{\mathcal X}$ and the signed imbalance $\Delta_{\mathcal X}$.

\begin{remark}\label{rem:delta_flux}
    Under the regularity conditions permitting differentiation of the marginal
    probability,
    \begin{equation}
        \Delta_{\mathcal X}(u;t)
        =\frac{\partial}{\partial t}\P(X(t)>u),
        \label{eq:delta_cdf_flux}
    \end{equation}
    so \(\Delta_{\mathcal X}\) is the probability flux through \(u\). It vanishes
    for time-invariant marginals. In particular, stationary processes satisfy
    \[\U_{\mathcal X}=\D_{\mathcal X}=\N_{\mathcal X}/2.\]
\end{remark}

The following weak but simple lower bound is useful in drift-dominated regimes and ensures that \eqref{eq:U_from_NDelta} and \eqref{eq:D_from_NDelta} make sense.
\begin{proposition}
Let $\mathcal X$ be a Rice-regular process. Then,
\begin{align}
    \N_{\mathcal X}(u)
    &\ge
    \left|\Delta_{\mathcal X}(u)\right|, \quad u>0.
    \label{eq:flux_lower_bound}
\end{align}
\end{proposition}

\begin{proof}
Apply Jensen's inequality \cite{billingsley1986probability} to \(\E[|\dot X|\given X=u]\).
\end{proof}

\section{Universal LCR Bounds for the Product}
\label{sec:universal_bounds}

Let $\mathcal R_1,\ldots,\mathcal R_n$ be positive, mutually independent,
Rice-regular processes defined over the same time interval $T$ of $\R$.
Let $\mathcal R$ be the product of $\mathcal R_1, \dots, \mathcal R_n$, that is,
\begin{align*}
    R(t) \triangleq R_1(t) \cdots R_n(t).
\end{align*}

For $k=1,\ldots,n$, define $f_k\triangleq f_{R_k}$, $\N_k \triangleq \N_{\mathcal R_k}$, $\Delta_k \triangleq \Delta_{\mathcal R_k}$ and the product of all factors except the $k$th one as
\begin{equation}
    C_k(t)\triangleq \prod_{i\ne k}R_i(t).
\end{equation}
Thus, $R(t)=C_k(t)R_k(t)$ and $C_k(t)>0$.

For admissible $u$, define the averaged level-$u$ contribution of the $k$th factor to the total LCR and to the signed crossing imbalance of the product as
\begin{align}
    \overline{\N}_k(u;t)
    &\triangleq
    \E\!\left[
    C_k(t)\left|\dot R_k(t)\right|
    \given R(t)=u
    \right]f_{R(t)}(u), \label{eq:Nk}\\
    \overline{\Delta}_k(u;t)
    &\triangleq
    \E\!\left[
    C_k(t)\dot R_k(t)
    \given R(t)=u
    \right]f_{R(t)}(u), \label{eq:deltak}
\end{align}
respectively. The overline in $\overline{\N}_k$ and $\overline{\Delta}_k$ denotes averaging over the remaining factors; it is not a time average.

An equivalent definition, by independence, is
\begin{align}
    \overline{\N}_k(u;t)
    &=
    \E\!\left[
    \N_k\!\left(\frac{u}{C_k(t)};t\right)
    \right], \\
    \overline{\Delta}_k(u;t)
    &=
    \E\!\left[
    \Delta_k\!\left(\frac{u}{C_k(t)};t\right)
    \right].
\end{align}

From now on, all statements are pointwise in an admissible time $t$. Thus, when no ambiguity
arises, $t$ is omitted.

By sample-path differentiability, the product rule gives
\begin{align} \label{eq:Rdot}
    \dot R
    =
    \sum_{k=1}^n C_k\dot R_k.
\end{align}

Consequently, the signed crossing imbalance of the product
decomposes exactly into the sum of the averaged signed
contributions. Thus, from \eqref{eq:delta_X},
\begin{align} \label{eq:DeltaR}
    \Delta_{\mathcal R}(u)
    &=
    \sum_{k=1}^n \overline{\Delta}_k(u),
\end{align}
by linearity of conditional expectation.

Next, we present the theorems that provide universal lower and upper bounds for the LCR, built upon the LCRs of the constituent processes.
\begin{theorem} \label{th:general_upper_bound}
Let $\mathcal R_1, \dots, \mathcal R_n$ be independent Rice-regular processes.
If their product $\mathcal R$ is Rice-regular, then
    \begin{align*}
        \N_{\mathcal R}(u)
        &\le \sum_{k=1}^n \overline{\N}_k(u), \quad u>0.
    \end{align*}
\end{theorem}

\begin{proof}
Using the triangle inequality in \eqref{eq:Rdot},
\[
    |\dot R|
    =
    \left|
    \sum_{k=1}^n C_k\dot R_k
    \right|
    \le
    \sum_{k=1}^n C_k|\dot R_k|,
\]
because $C_k>0$. Hence,
\begin{align}
    \N_{\mathcal R}(u)
    &\le
    \sum_{k=1}^n
    \E\!\left[
    C_k|\dot R_k|\given R=u
    \right]f_{R}(u) \nonumber
    =
    \sum_{k=1}^n \overline{\N}_k(u).
\end{align}
\end{proof}

\begin{theorem} \label{th:general_lower_bound}
Let $\mathcal R_1, \dots, \mathcal R_n$ be independent Rice-regular processes.
If their product  $\mathcal R$ is Rice-regular, then
    \begin{align*}
        \N_{\mathcal R}(u)
        &\ge \max_{1\le k \le n} 
        \left\{\overline{\N}_k(u) + d_k(u) \right\}, \quad u>0,
    \end{align*}
    where $d_k$ is the drift correction term
    \begin{align*}
        d_k(u)
        \triangleq
        \sum_{i\ne k}
        \E\!\left[
        C_i\dot R_i\operatorname{sgn}(\dot R_k)
        \given R=u
        \right]f_{R}(u).
    \end{align*}
\end{theorem}

\begin{proof}
Fix $k$ and write
\[
    \dot R
    =
    C_k\dot R_k
    +
    \sum_{i\ne k} C_i\dot R_i.
\]
Using
$
    |a+b|\ge |a|+b\operatorname{sgn}(a),
$
we obtain
\begin{align*}
    |\dot R|
    &\ge
    C_k|\dot R_k|
    +
    \sum_{i\ne k}
    C_i\dot R_i\operatorname{sgn}(\dot R_k).
\end{align*}
Taking conditional expectation given $R=u$ and multiplying by
$f_{R}(u)$ gives
$
    \N_{\mathcal R}(u)
    \ge
    \overline{\N}_k(u)+d_k(u).
$
Since this inequality holds for every $k$, maximizing over $k$ completes the proof.
\end{proof}

\begin{remark}
\label{rem:dominant_LCR_component}
The universal bounds become exact when one averaged component dominates the
others. This follows from
\[|d_k(u)| \le \sum_{i\neq k} \E\!\left[
        C_i |\dot R_i|
        \given R=u
        \right]f_{R}(u)
        = \sum_{i\neq k}\overline\N_i(u).\]
For a given $k$, if $\overline\N_i/\overline\N_k \to 0$ for every $i\neq k$, then the ratio of the bounds tends to 1.
\end{remark}

In Theorem~\ref{th:general_lower_bound}, the correction $d_k$ is the price of nonstationarity in the factor-wise lower bound. In general, one cannot replace the lower bound in Theorem~\ref{th:general_lower_bound} by $\max_k\overline{\N}_k$, because the drift of the remaining factors may either reinforce or cancel the motion of the $k$th factor at the crossing level.
On the other hand, if the processes are stationary the replacement can be made, which results in the following corollary.

\begin{corollary} \label{cor:general_stat_bounds}
    Let $\mathcal R_1, \dots, \mathcal R_n$ be independent stationary Rice-regular processes.
    If their product $\mathcal R$ is Rice-regular, then
    \begin{align*}
         \left\lVert \overline{\mathbf N}(u) \right\rVert_\infty \le \N_{\mathcal R}(u) \le \left\lVert \overline{\mathbf N}(u) \right\rVert_1,
    \end{align*}
    for every $u>0$, where $\overline{\mathbf N}(u) \triangleq \bigl(\overline\N_1(u), \dots, \overline\N_n(u)\bigr).$
\end{corollary}
\begin{proof}
    Conditioning on \(\mathbf R=(R_1,\ldots,R_n)\),
    \[
    \E[C_i\dot R_i\operatorname{sgn}(\dot R_k)\given\mathbf R]
    =C_i\E[\dot R_i\given R_i]\E[\operatorname{sgn}(\dot R_k)\given R_k]
    \]
    by independence. Stationarity gives \(\E[\dot R_i\given R_i]=0\), hence \(d_k=0\). Applying Theorems~\ref{th:general_upper_bound} and~\ref{th:general_lower_bound} ends the proof.
\end{proof}

Thus, the general nonstationary bounds recover the general stationary product bounds in~\cite{dester2025level}. In fact, the upper bound is identical.

    
    
    
    
    
    
    
    

\section{Bounds and Exact Form for DGD Products}
\label{sec:DGD_Bound}

We now refine the previous product bounds under a conditional
Gaussian-derivative assumption. This assumption is satisfied under the
standard temporal constructions of several classical fading and shadowing
models. For example, the conventional isotropic \(\alpha\)-\(\mu\)
construction yields a conditionally Gaussian derivative and includes
Rayleigh, Nakagami-\(m\), Weibull, exponential, and related models as
special cases. Rice and \(\kappa\)-\(\mu\) fading also satisfy the
assumption under their standard isotropic Gaussian-component
constructions, while lognormal shadowing satisfies it when generated as
the exponential of a differentiable Gaussian process
\cite{yacoub2007alpha,yacoub2007kappa,patzold2010shadow}.
Next, we present the definition of a DGD process.

\subsection{DGD Process Definition}

\begin{definition} \label{def:DGD}
    A Rice-regular process \(\X\) belongs to the
    \textit{drifting Gaussian-derivative} class if, for every admissible
    \(t\) and \(x>0\),
    \[
        \dot X(t)\given (X(t)=x)
        \sim
        \mathcal N\!\left(\mu_\X(x;t),\sigma_\X^2(x;t)\right),
    \]
    where
    \begin{align*}
        \mu_\X(x;t)
        &\triangleq
        \E[\dot X(t)\given X(t)=x],\\
        \sigma_\X^2(x;t)
        &\triangleq
        \V(\dot X(t)\given X(t)=x) > 0.
    \end{align*}
\end{definition}

The following function is useful for LCR calculations of DGD processes.
For \(a\in\mathbb R\) and \(b>0\), define
\begin{align}
    \Psi(a,b)
    \triangleq
    \E[|a+bZ|],
    \qquad Z\sim\mathcal N(0,1).
\end{align}
Straightforward calculation yields
\begin{align}
    \Psi(a,b)
    =
    2b\phi\!\left(\frac{a}{b}\right)
    +
    a
    \left[
    2\Phi\!\left(\frac{a}{b}\right)-1
    \right],
\end{align}
where \(\phi\) and \(\Phi\) denote the standard Gaussian PDF and CDF, respectively.
The following elementary properties of \(\Psi\) will be used in the derivation of the DGD bounds.

\begin{lemma}\label{lem:Psi_properties}
    The function \(\Psi\) satisfies the following properties:
    \begin{enumerate}
        \item \(\Psi\) is even in its first argument, i.e.,
        \(\Psi(a,b)=\Psi(-a,b)\).
        \item \(\Psi\) is positively homogeneous, i.e.,
        \[
            \Psi(ca,cb)=c\Psi(a,b), \qquad c > 0.
        \]
        \item \(\Psi\) is convex on \(\R\times\R_{>0}\).
        \item For fixed \(a\), \(\Psi(a,b)\) is increasing in \(b>0\).
    \end{enumerate}
\end{lemma}

\begin{proof}
Evenness and homogeneity follow from symmetry of \(Z\); convexity follows
because \((a,b)\mapsto|a+bZ|\) is convex; and
\(\partial\Psi(a,b)/\partial b=2\phi(a/b)>0\).
\end{proof}

For a single DGD process, the total LCR and the signed imbalance are obtained directly from the first absolute moment and the mean of the conditional Gaussian derivative as follows.

\begin{lemma} \label{lem:DGD}
    If $\X$ is a DGD process, then, for $u>0$,
    \begin{align*}
        \Delta_\X(u)
        &=
        \mu_\X(u)f_X(u),\\
        \N_\X(u)
        &=
        \Psi(\mu_\X(u),\sigma_\X(u)) f_X(u),
    \end{align*}
\end{lemma}
\begin{proof}
    Follows directly from \eqref{eq:N_def} and \eqref{eq:delta_X}.
\end{proof}

\subsection{Product of DGD Processes}

Let $\mathcal R_1,\ldots,\mathcal R_n$ be DGD and $\mathcal R$ their product.
Define
\begin{align}
    \mu_k(r) &\triangleq \E[\dot R_k\given R_k=r],\\
    \sigma^2_k(r) &\triangleq \V(\dot R_k\given R_k=r),
\end{align}
for $k=1,\ldots,n$ and $r>0$, and
\begin{align}
    M_{\mathcal R} &\triangleq \sum_{k=1}^n C_k\mu_k(R_k),\\
    \mathbf S_{\mathcal R} &\triangleq \bigl(C_1\sigma_1(R_1),\ldots,C_n\sigma_n(R_n)\bigr).
\end{align}

Then, a formula for the crossing imbalance $\Delta_{\mathcal R}$ is readily derived from \eqref{eq:DeltaR} and Lemma~\ref{lem:DGD} as
\begin{align} \label{eq:DGD_DeltaR}
    \Delta_{\mathcal R}(u)
    &= \E[M_{\mathcal R}\given R=u] f_R(u) \\
    &=
    \sum_{k=1}^n
    \E\!\left[
    \mu_k\!\left(\frac{u}{C_k}\right)f_{k}\!\left(\frac{u}{C_k}\right)
    \right].
\end{align}

We next characterize the product process. Although the product of DGD factors is not necessarily DGD, its conditional derivative admits an exact Gaussian mixture representation.

\begin{lemma}\label{lem:DGD_exact}
    Let $\mathcal R_1, \dots, \mathcal R_n$ be independent DGD processes.
    If their product $\mathcal R$ is Rice-regular, then
    \begin{align*}
        \N_{\mathcal R}(u)
        &= \E[\Psi(M_{\mathcal R},\|\mathbf S_{\mathcal R}\|_2)\given R=u]f_R(u), \quad u>0.
    \end{align*}
\end{lemma}
\begin{proof}
    Since the processes \(\mathcal R_1,\ldots,\mathcal R_n\) are mutually
    independent, the random vectors \((R_k,\dot R_k)\), \(k=1,\ldots,n\), are
    mutually independent. Hence, conditionally on
    \(\mathbf R=(R_1,\ldots,R_n)\), the random variables
    \(\dot R_1,\ldots,\dot R_n\) are independent.
    Moreover, by the DGD assumption,
    \[
        \dot R_k \given \mathbf R
        \sim
        \mathcal N\!\left(\mu_k(R_k),\sigma_k^2(R_k)\right),
        \qquad k=1,\ldots,n .
    \]
    Using the product rule,
    $
        \dot R=\sum_{k=1}^n C_k\dot R_k .
    $
    Hence, 
    \[
        \dot R \given \mathbf R \sim \mathcal N(M_{\mathcal R},\|\mathbf S_{\mathcal R}\|_2^2).
    \]
    Therefore,
    $
        \E\!\left[|\dot R|\given \mathbf R\right]
        =
        \Psi(M_{\mathcal R},\|\mathbf S_{\mathcal R}\|_2).
    $
    Since $R$ is a function of $\mathbf R$, the total expectation law gives
    \begin{align*}
        \E\!\left[|\dot R|\given R=u\right]
        &=
        \E\!\left[
            \E\!\left[|\dot R|\given \mathbf R\right]
            \given R=u
        \right] \\
        &=
        \E\!\left[\Psi(M_{\mathcal R},\|\mathbf S_{\mathcal R}\|_2)\given R=u\right].
    \end{align*}
    Multiplying by $f_R(u)$ and applying \eqref{eq:N_def} proves the claim.
\end{proof}


\begin{theorem} \label{th:DGD_lower_bound}
    Let $\mathcal R_1, \dots, \mathcal R_n$ be independent DGD processes.
    If their product $\mathcal R$ is Rice-regular, then
    \begin{align*}
        \N_{\mathcal R}(u) 
        &\ge
        \Psi\!\left(
        \Delta_{\mathcal R}(u),
        \left\lVert \overline{\boldsymbol\Sigma}_{\mathcal R}(u)\right\rVert_2
        \right), \quad u>0,
    \end{align*}
    where
    $
        \overline{\boldsymbol\Sigma}_{\mathcal R}(u)
        \triangleq
        \E[\mathbf S_\mathcal R\given R=u] f_R(u).
    $
\end{theorem}
\begin{proof}
    By Lemma~\ref{lem:DGD_exact},
    \[
        \N_{\mathcal R}(u)
        =
        f_R(u)\E[\Psi(M_{\mathcal R},\|\mathbf S_{\mathcal R}\|_2)\given R=u].
    \]
    Using the absolute homogeneity and convexity of $\Psi$, Jensen's inequality gives
    \begin{align*}
        \N_{\mathcal R}(u)
        &=
        \E\!\left[
            \Psi\!\left(f_R(u)M_{\mathcal R},f_R(u)\|\mathbf S_{\mathcal R}\|_2\right)
            \given R=u
        \right] \\
        &\ge
        \Psi\!\left(
            f_R(u)\E[M_{\mathcal R}\given R=u],
            f_R(u)\E[\|\mathbf S_{\mathcal R}\|_2\given R=u]
        \right).
    \end{align*}
    The first argument is $\Delta_{\mathcal R}(u)$
    by \eqref{eq:DeltaR}. 
    
    By Jensen's inequality for the $\ell_2$-norm,
    \begin{align*}
        f_R(u)\E[\|\mathbf S_{\mathcal R}\|_2\given R=u]
        &\ge
        \left\|
            f_R(u)\E[\mathbf S_{\mathcal R}\given R=u]
        \right\|_2 \\
        &=\|\overline{\boldsymbol\Sigma}_{\mathcal R}(u)\|_2.
    \end{align*}
    Since $\Psi(a,b)$ is nondecreasing in $b\ge0$, it follows that
    \[
        \N_{\mathcal R}(u)
        \ge
        \Psi\!\left(
            \Delta_{\mathcal R}(u),
            \|\overline{\boldsymbol\Sigma}_{\mathcal R}(u)\|_2
        \right),
    \]
    which proves the result.
\end{proof}

To obtain a tractable upper bound, we first derive a sharp quadratic upper bound for $\Psi$ using Sampford's inequality \cite{sampford1953mills}.

\begin{lemma} \label{lem:Psi_upper_bound}
    For every \(a\in\R\) and \(b > 0\),
    \[
        \Psi(a,b)
        \le
        \sqrt{a^2+\frac{2}{\pi}b^2}.
    \]
\end{lemma}

\begin{proof}
    See Appendix~\ref{app:proof_upper_bound_Psi}.
\end{proof}

The bound is exact at \(a=0\) and asymptotically tight in relative error
as \(|a|/b\to\infty\), since
\[
\sqrt{a^2+\tfrac{2}{\pi}b^2}-\Psi(a,b)
=\frac{b^2}{\pi|a|}+\mathcal O\!\left(\frac{b^4}{|a|^3}\right).
\]

The inequality yields the following upper bound.

\begin{theorem} \label{th:DGD_upper_bound}
    Let $\mathcal R_1, \dots, \mathcal R_n$ be independent DGD processes.
    If their product $\mathcal R$ is Rice-regular, then, for $u > 0$,
    \[
        \N_{\mathcal R}(u)
        \le
        \sqrt{
        \E\!\left[
        M_{\mathcal R}^2+\frac{2}{\pi}\|\mathbf S_{\mathcal R}\|_2^2
        \given R=u
        \right]
        } f_R(u).
    \]
    
\end{theorem}
\begin{proof}
    Lemma~\ref{lem:Psi_upper_bound} and Jensen's inequality for the concave
    square root give
    \begin{align*}
    \E\!\left[
        \Psi(M_{\mathcal R},\|\mathbf S_{\mathcal R}\|_2)
        \given R=u
    \right]
    &\le
    \E\!\left[
        \sqrt{
            M_{\mathcal R}^2+
            \tfrac{2}{\pi}\|\mathbf S_{\mathcal R}\|_2^2
        }
        \given R=u
    \right]\\
    &\le
    \sqrt{
        \E\!\left[
            M_{\mathcal R}^2+
            \tfrac{2}{\pi}\|\mathbf S_{\mathcal R}\|_2^2
            \given R=u
        \right]
    }.
    \end{align*}
    Multiplying by \(f_R(u)\) ends the proof.
\end{proof}


\begin{remark} \label{rem:DGD_bounds}
The DGD bounds become asymptotically equivalent in two important
regimes. 

First, in a drift-dominated regime, if
\begin{equation}
    \frac{
        \V(M_{\mathcal R}\given R=u)
        +
        \E[\|\mathbf S_{\mathcal R}\|_2^2\given R=u]
    }{
        \E[M_{\mathcal R}\given R=u]^2
    }
    \to 0,
\end{equation}
then both bounds are asymptotic to
$
    |\Delta_{\mathcal R}(u)|.
$

Second, in a spread-dominated regime, if
\[
    \frac{
        \E[M_{\mathcal R}^2\given R=u]
    }{
        \|\E[\mathbf S_{\mathcal R}\given R=u]\|_2^2
    }
    \to 0
\]
and
\[
    \frac{
        \E\!\left[
            \|\mathbf S_{\mathcal R}
            -\E[\mathbf S_{\mathcal R}\given R=u]\|_2^2
            \given R=u
        \right]
    }{
        \|\E[\mathbf S_{\mathcal R}\given R=u]\|_2^2
    }
    \to 0,
\]
then both DGD bounds are asymptotic to
$
    \sqrt{\frac{2}{\pi}}\,
    \|\overline{\boldsymbol\Sigma}_{\mathcal R}(u)\|_2.
$
\end{remark}

In the stationary case, the drift terms vanish and the preceding bounds simplify considerably.

\begin{corollary}\label{cor:DGD_stat_bounds}
    Let $\mathcal R_1, \dots, \mathcal R_n$ be independent stationary DGD processes.
    If their product $\mathcal R$ is Rice-regular, then
    \[
        \left\lVert \overline{\mathbf N}(u) \right\rVert_2
        \le
        \N_{\mathcal R}(u)
        \le
         \sqrt{\frac{2}{\pi} \E[\|\mathbf S_{\mathcal R}\|_2^2 \given R=u] } \, f_R(u),
    \]
    for every $u>0$, where
    $
        \overline{\mathbf N}(u)
        =
        \bigl(\overline{\N}_1(u),\ldots,\overline{\N}_n(u)\bigr).
    $
\end{corollary}
\begin{proof}
Stationarity gives \(\mu_k(r)=0\) at every nontrivial level in the support of \(R_k\).
Consequently, \(M_{\mathcal R}=0\) and \(\Delta_{\mathcal R}(u)=0\).
Moreover, by Lemma~\ref{lem:DGD},
$
    \overline{\mathbf N}(u)
    =
    \sqrt{\frac{2}{\pi}}\,
    \overline{\boldsymbol\Sigma}_{\mathcal R}(u).
$
Substitution in Theorems~\ref{th:DGD_lower_bound}
and~\ref{th:DGD_upper_bound} yields the result.
\end{proof}

Thus, the DGD nonstationary lower bound recovers the Gaussian-derivative stationary bound in~\cite{dester2025level}.
On the other hand, to the best of the author's knowledge, the upper bound is novel, even for the stationary case.

\section{Applications}
\label{sec:application}

\subsection{Nonstationary Double-Rayleigh}
\label{subsec:nonstationary_double_rayleigh}

Consider the nonstationary Double-Rayleigh product
\[
    R(t)=R_1(t)R_2(t),
\]
where
\[
    R_k(t)=a_k(t)Z_k(t),\qquad k=1,2.
\]
Here, \(a_k(t)>0\) is a deterministic continuously differentiable
large-scale gain, and \(Z_k(t)\) is a normalized stationary Rayleigh envelope
with density
\[
    f_{Z_k}(z)=2z\euler^{-z^2},\qquad z>0.
\]
Let \(\omega_k\) denote the angular maximum Doppler shift of \(Z_k\).
Then
\[
    \dot Z_k\given (Z_k=z)
    \sim
    \mathcal N\!\left(0,\frac{\omega_k^2}{4}\right).
\]
The stationary Rayleigh processes \(Z_1\) and \(Z_2\) are independent.

Since
\[
    \dot R_k(t)
    =
    \dot a_k(t)Z_k(t)+a_k(t)\dot Z_k(t),
\]
we have, conditionally on \(R_k(t)=r\),
\[
    \dot R_k(t)\given (R_k(t)=r)
    \sim
    \mathcal N\!\left(
    \lambda_k(t)r,\,
    a_k^2(t)\frac{\omega_k^2}{4}
    \right),
\]
where
$
    \lambda_k(t)\triangleq \frac{\dot a_k(t)}{a_k(t)} .
$
Hence each \(\mathcal R_k\) is DGD with
\[
    \mu_k(r;t)=\lambda_k(t)r,
    \qquad
    \sigma_k(r;t)=a_k(t)\frac{\omega_k}{2}.
\]

Let the normalized level and the aggregate drift rate be
\[
    \rho(t)\triangleq \frac{u}{a_1(t)a_2(t)},
    \qquad
    \lambda_\Sigma(t)\triangleq \lambda_1(t)+\lambda_2(t).
\]
In what follows, we omit the explicit dependence on \(t\). By
Lemma~\ref{lem:DGD_exact}, the exact total LCR of \(\mathcal R\) at level
\(u>0\) is
\begin{align}
    \N_{\mathcal R}(u)
    &=
    4\rho
    \int_0^\infty
    \frac{1}{x}
    \Psi\!\left(
    \rho\lambda_\Sigma,
    \frac{1}{2}
    \sqrt{
    \omega_1^2x^2+\omega_2^2\frac{\rho^2}{x^2}
    }
    \right) \nonumber\\
    &\qquad\qquad\cdot
    \exp\!\left(
    -x^2-\frac{\rho^2}{x^2}
    \right)\d x .
    \label{eq:ns_double_rayleigh_exact}
\end{align}
A closed form is recovered only when \(\lambda_\Sigma=0\) and \(\omega_1=\omega_2\) \cite[Appendix~A]{dester2025level}, which includes the stationary i.i.d. case.

The averaged individual contributions are
\begin{align}
    \overline{\N}_k(u)
    &=
    4\rho
    \int_0^\infty
    \Psi\!\left(
    \lambda_k\frac{\rho}{x},
    \frac{\omega_k}{2}
    \right)
    \exp\!\left(
    -x^2-\frac{\rho^2}{x^2}
    \right)\d x .
\end{align}

We define the lower incomplete modified Bessel function of the second kind and
order zero by \cite{jones2007incomplete}
\[
    \operatorname{K}_0^{\rm L}(z,a)
    \triangleq
    \int_0^a e^{-z\cosh y}\,\d y ,
    \qquad z>0,\ a\ge0.
\]
Then, for \(k=1,2\), the averaged individual contributions admit the closed
forms derived in Appendix~\ref{app:Nk}:
\begin{align}
    \overline{\N}_k(u)
    &=
    \sqrt{2}\,\omega_k\rho
    \exp\!\left[
    -2\rho
    \sqrt{
    1+\frac{2\lambda_k^2}{\omega_k^2}
    }
    \right] \nonumber\\
    &\quad+
    4|\lambda_k|\rho^2
    \operatorname{K}_0^{\rm L}\!\left(
    2\rho,
    \operatorname{arsinh}
    \frac{\sqrt{2}|\lambda_k|}{\omega_k}
    \right).
    \label{eq:Nbar_closed}
\end{align}
Similarly, for \((k,\overline k)=(1,2),(2,1)\), we obtain
\begin{align}
    d_k(u)
    &=
    4\lambda_{\overline k}\rho^2
    \operatorname{sgn}(\lambda_k)
    \operatorname{K}_0^{\rm L}\!\left(
    2\rho,
    \operatorname{arsinh}
    \frac{\sqrt{2}|\lambda_k|}{\omega_k}
    \right).
\end{align}
Then the universal bounds of Theorems~\ref{th:general_upper_bound}
and~\ref{th:general_lower_bound} are
\begin{equation}
    \max_{k=1,2}\{\overline{\N}_k(u)+d_k(u)\}
    \le
    \N_{\mathcal R}(u)
    \le
    \overline{\N}_1(u)+\overline{\N}_2(u).
    \label{eq:ns_double_rayleigh_universal_bounds}
\end{equation}

Let
\[
    \boldsymbol\omega\triangleq(\omega_1,\omega_2),
    \qquad
    \|\boldsymbol\omega\|_2=\sqrt{\omega_1^2+\omega_2^2}.
\]
For the DGD bounds of Theorems~\ref{th:DGD_lower_bound}
and~\ref{th:DGD_upper_bound}, first note that
\[
    \left\lVert
    \overline{\boldsymbol\Sigma}_{\mathcal R}(u)
    \right\rVert_2
    =
    \sqrt{\pi}\,
    \|\boldsymbol\omega\|_2
    \rho\euler^{-2\rho}.
\]
Moreover, \(M_{\mathcal R}=\rho\lambda_\Sigma\) on the product level set \(Z_1Z_2=\rho\), and
hence
\[
    \Delta_{\mathcal R}(u)
    =
    4\rho^2\lambda_\Sigma \operatorname{K}_0(2\rho),
\]
where \(\operatorname{K}_\nu\) is the modified Bessel function of the second
kind and order \(\nu\) \cite[Eq.~(7.2.13)]{bateman1953higher}. Therefore, using the homogeneity of \(\Psi\), the
DGD lower bound becomes
\begin{align}
    \frac{\N_{\mathcal R}(u)}{\|\boldsymbol\omega\|_2}
    &\ge
    4\rho^2 \operatorname{K}_0(2\rho) \, \Psi\!\left(
    \frac{\lambda_\Sigma}{\|\boldsymbol\omega\|_2} ,
    \frac{\sqrt{\pi}\,
    \euler^{-2\rho}}{4\rho\operatorname{K}_0(2\rho)}
    \right).
    \label{eq:ns_double_rayleigh_DGD_lower}
\end{align}

Since
\[
    \E[\|\mathbf S_{\mathcal R}\|_2^2\given Z_1Z_2=\rho]
    =
    \rho
    \frac{\|\boldsymbol\omega\|_2^2}{4}
    \frac{\operatorname{K}_1(2\rho)}{\operatorname{K}_0(2\rho)},
\]
Theorem~\ref{th:DGD_upper_bound} gives
\begin{align}
    \frac{\N_{\mathcal R}(u)}{\|\boldsymbol\omega\|_2}
    &\le
    4\rho^2 \operatorname{K}_0(2\rho)
    \sqrt{
        \dfrac{\lambda_\Sigma^2}{\|\boldsymbol\omega\|_2^2}
        +
        \dfrac{\operatorname{K}_1(2\rho)}{2\pi\rho\operatorname{K}_0(2\rho)}
    }.
    \label{eq:ns_double_rayleigh_DGD_upper}
\end{align}

The scale factors affect the LCR through \(\rho(t)\) and the aggregate
drift \(\lambda_\Sigma(t)\). Figure~\ref{fig:double_rayleigh_lcr_bounds}
shows that, for the considered parameters, the DGD bounds are markedly tighter
than the universal bounds.

\begin{figure}
    \centering
    \includegraphics[width=1.0\linewidth]{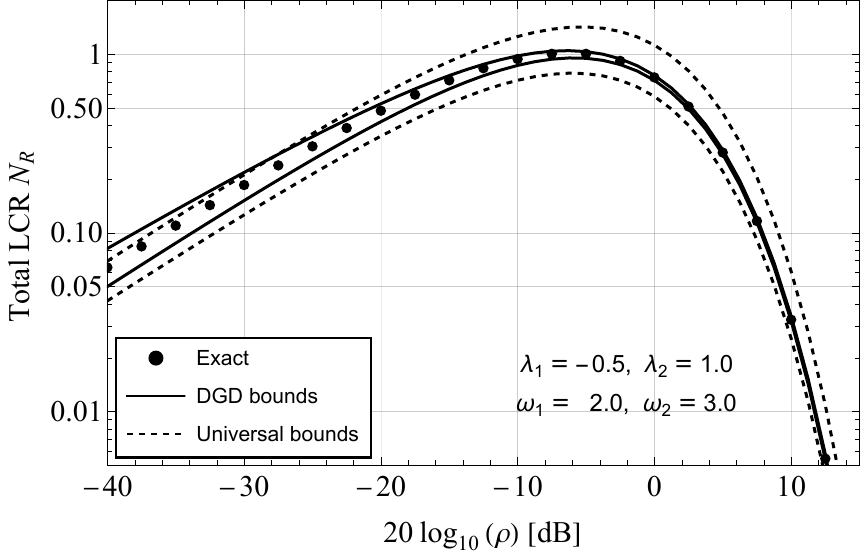}
    \caption{Exact LCR and bounds for Double-Rayleigh fading with
    \(\lambda_1=-0.5\), \(\lambda_2=1\), \(\omega_1=2\), and
    \(\omega_2=3\) (in units of [time]$^{-1}$).}
    \label{fig:double_rayleigh_lcr_bounds}
\end{figure}

A useful way to quantify the tightness of the DGD bounds is through their
relative gap. Denote by \(\mathfrak{L}_{\rm DGD}\) and \(\mathfrak{U}_{\rm DGD}\) the lower and upper bounds in \eqref{eq:ns_double_rayleigh_DGD_lower} and
\eqref{eq:ns_double_rayleigh_DGD_upper}, respectively. Using the positive
homogeneity of \(\Psi\), we obtain,
for every fixed \({\rho>0}\), as \(|\lambda_\Sigma|/\|\boldsymbol\omega\|_2 \to \infty\),
\begin{align}
    \frac{\mathfrak U_{\rm DGD}}{\mathfrak L_{\rm DGD}}
    &=
    1+
    \frac{\operatorname K_1(2\rho)}
    {4\pi\rho\operatorname K_0(2\rho)}
    \frac{\|\boldsymbol\omega\|_2^2}{\lambda_\Sigma^2}
    +
    \mathcal O\!\left(
        \frac{\|\boldsymbol\omega\|_2^4}{\lambda_\Sigma^4}
    \right).
\end{align}
Thus, the DGD bounds are asymptotically tight in the drift-dominated regime, as expected from Remark~\ref{rem:DGD_bounds}.

On the other hand, in the spread-dominated regime, as \(|\lambda_\Sigma|/\|\boldsymbol\omega\|_2 \to 0 \), we have that
\begin{align}
    \frac{\mathfrak U_{\rm DGD}}{\mathfrak L_{\rm DGD}}
    &=
    \frac{2 \euler^{2\rho}}{\sqrt{\pi}}
    \sqrt{\rho \operatorname K_0(2\rho)\operatorname K_1(2\rho)}
    +
    \mathcal O\!\left(
        \frac{\lambda_\Sigma^2}{\|\boldsymbol\omega\|_2^2}
    \right).
\end{align}
Using \cite[Secs. 10.29 and 10.32]{DLMF}, we obtain that
\[
    \frac{2 \euler^{2\rho}}{\sqrt{\pi}}
    \sqrt{\rho \operatorname K_0(2\rho)\operatorname K_1(2\rho)}
    < \sqrt{1 + \frac{1}{8\rho}}.
\]
Thus, the bounds yield tight approximations for large $\rho$ in the spread-dominated regime.

\begin{figure}
    \centering
    \includegraphics[width=1.0\linewidth]{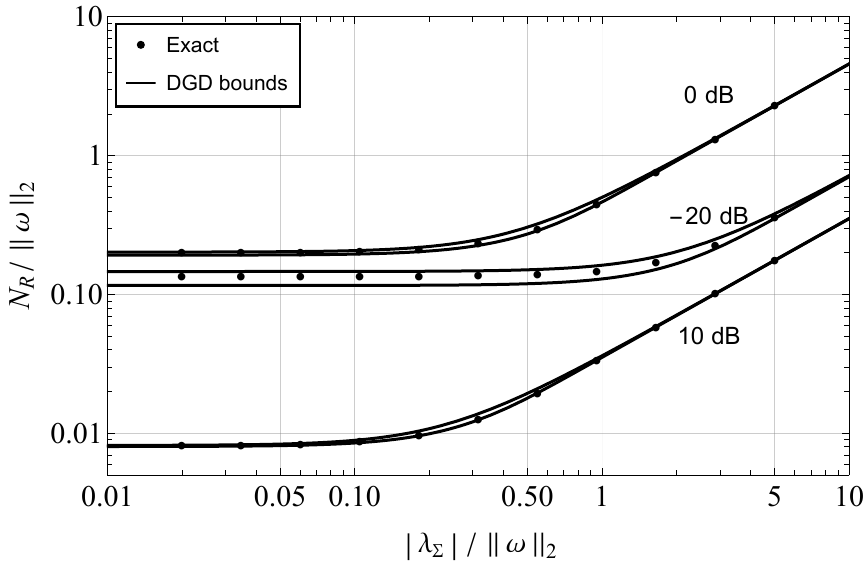}
    \caption{Normalized exact LCR and DGD bounds versus
    \(|\lambda_\Sigma|/\|\boldsymbol\omega\|_2\) for
    \(20\log_{10}\rho=-20,0,10\) dB and the symmetric Doppler split
    \(\omega_1=\omega_2\).} 
    \label{fig:double_rayleigh_lcr_ratio}
\end{figure}

Figure~\ref{fig:double_rayleigh_lcr_ratio} confirms that the exact LCR lies
between the DGD bounds and that the enclosure collapses for fixed \(\rho\) as \(|\lambda_\Sigma|/\|\boldsymbol\omega\|_2\to\infty\) and for large $\rho$ as \(|\lambda_\Sigma|/\|\boldsymbol\omega\|_2\to0\), consistently with the preceding expansions.

\subsection{Quasi-Identically Distributed DGD Products}
\label{subsec:quasi_identical_DGD}

Let
\(\mathcal R_1,\ldots,\mathcal R_n\) be independent DGD processes such that,
at a fixed admissible time \(t\), the envelopes \(R_1,\ldots,R_n\) are
identically distributed with common density \(f_0\). Assume further that their
conditional derivative parameters have the linear-drift common-shape form
\[
    \mu_k(r)=\lambda_k r,
    \qquad
    \sigma_k(r)=\omega_k s_0(r),
    \qquad k=1,\ldots,n,
\]
where \(s_0(r)>0\) is common to all factors, whereas
\(\lambda_k\in\mathbb R\) and \(\omega_k>0\) quantify, respectively, the
drift rate and the angular maximum Doppler shift of the \(k\)th factor.
Let
\[
    \lambda_\Sigma
    \triangleq
    \sum_{k=1}^n\lambda_k,
    \qquad
    \boldsymbol\omega
    \triangleq
    (\omega_1,\ldots,\omega_n).
\]
Since \(R=C_kR_k\), on the product level set \(R=u\) we have
$
    C_kR_k=u,~ k=1,\ldots,n.
$
Consequently, the conditional drift of \(\dot R\) on the level set is
$
    M_{\mathcal R}
    =
    u\lambda_\Sigma.
$
Thus,
\[
    \Delta_{\mathcal R}(u)
    =
    u f_R(u)\lambda_\Sigma .
\]

By exchangeability, define the first-order averaged spread contribution
\[
    \Sigma_0(u)
    \triangleq
    \E\!\left[
        s_0\!\left(\frac{u}{C_1}\right)
        f_0\!\left(\frac{u}{C_1}\right)
    \right].
\]
Then,
\[
    \left\lVert\overline{\boldsymbol\Sigma}_{\mathcal R}(u)\right\rVert_2
    =
    \Sigma_0(u)\|\boldsymbol\omega\|_2.
\]
For the DGD upper bound, define the second-order averaged spread contribution
\begin{align}
    \Theta_0(u)
    &\triangleq
    f_R(u)
    \sqrt{
    \E\!\left[
        C_1^2s_0^2(R_1)
        \given R=u
    \right]
    }\\
    &=
    \sqrt{f_R(u)
    \E\!\left[
        C_1
        s_0^2\!\left(\frac{u}{C_1}\right)
        f_0\!\left(\frac{u}{C_1}\right)
    \right].}
\end{align}
Since
$
    \|\mathbf S_{\mathcal R}\|_2^2
    =
    \sum_{k} C_k^2\sigma_k^2(R_k),
$
exchangeability yields
\[
    f_R^2(u)\E[\|\mathbf S_{\mathcal R}\|_2^2\given R=u]
    =
    \Theta_0^2(u)\|\boldsymbol\omega\|_2^2.
\]

Finally, Theorems~\ref{th:DGD_lower_bound} and~\ref{th:DGD_upper_bound} give,
respectively,
\begin{align}
    \N_{\mathcal R}(u)
    &\ge
    \Psi\!\left(
        u f_R(u)\lambda_\Sigma,
        \Sigma_0(u)\|\boldsymbol\omega\|_2
    \right),
    \label{eq:quasi_DGD_linear_lower_bound}\\
    \N_{\mathcal R}(u)
    &\le
    \sqrt{
        u^2 f_R^2(u)\lambda_\Sigma^2
        +
        \tfrac{2}{\pi}
        \Theta_0^2(u)\|\boldsymbol\omega\|_2^2
    } .
    \label{eq:quasi_DGD_linear_upper_bound}
\end{align}

Thus, drifts combine algebraically through \(\lambda_\Sigma\), while
independent derivative spreads combine through \(\|\boldsymbol\omega\|_2\).

\begin{remark}
The exact LCR generally depends on the full allocation
\(\boldsymbol\omega\), whereas both DGD bounds depend only on
\(\|\boldsymbol\omega\|_2\). Hence, all allocations with the same Euclidean
norm share the same bounds, which are easier to interpret than the original LCR formulation.
\end{remark}

For the relative gap, let
\[
    \gamma
    \triangleq
    \frac{u f_R(u)|\lambda_\Sigma|}
    {\Sigma_0(u)\|\boldsymbol\omega\|_2},
    \qquad
    \eta
    \triangleq
    \frac{\Theta_0(u)}{\Sigma_0(u)}.
\]
By the Cauchy--Schwarz inequality, \(\eta\ge1\). Then
\begin{align}
    \frac{\mathfrak{U}_{\rm DGD}}{\mathfrak{L}_{\rm DGD}}
    &=
    \frac{
        \sqrt{
            \gamma^2+\frac{2}{\pi}\eta^2
        }
    }{
        \Psi(\gamma,1)
    }
    \label{eq:quasi_DGD_linear_gap}\\
    &=
    1+\frac{\eta^2}{\pi\gamma^2}
    +
    \mathcal{O}(\gamma^{-4}),
\end{align}
as \(\gamma\to\infty\). As \(\gamma\to0\),
\[
    \frac{\mathfrak{U}_{\rm DGD}}{\mathfrak{L}_{\rm DGD}}
    =
    \eta
    \left[
        1+
        \left(
            \frac{\pi}{4\eta^2}
            -
            \frac{1}{2}
        \right)\gamma^2
        +
        \mathcal{O}(\gamma^4)
    \right].
\]
Hence the gap tends to \(1\) in the drift-dominated regime and to
\(\eta\) in the spread-dominated regime. If \(s_0(r)\propto r\), then
\(\eta=1\).

\begin{table}[t]
\renewcommand{\arraystretch}{1.65}
\centering
\caption{Selected \(\alpha\)-\(\mu\) special cases.}
\label{tab:a-m_special_cases}
\begin{tabular}{|c|c|c|c|}\hline
Model & Parameters & \(f_0(r)\) & \(s_0(r)\) \\\hline
\(\alpha\)-\(\mu\) & \(\alpha,\mu>0\) &
\(\frac{\alpha\mu^\mu}{\Gamma(\mu)}r^{\alpha\mu-1}e^{-\mu r^\alpha}\) &
\(\frac{r^{1-\alpha/2}}{\alpha\sqrt\mu}\) \\\hline
Rayleigh & \(\alpha=2,\mu=1\) & \(2re^{-r^2}\) & \(1/2\) \\\hline
Nakagami-\(m\) & \(\alpha=2,\mu=m\) &
\(\frac{2m^m}{\Gamma(m)}r^{2m-1}e^{-mr^2}\) & \(1/(2\sqrt m)\) \\\hline
Weibull & \(\mu=1\) & \(\alpha r^{\alpha-1}e^{-r^\alpha}\) &
\(r^{1-\alpha/2}/\alpha\) \\\hline
\end{tabular}
\end{table}

\begin{figure*}[ht]
    \centering
    \subfloat[\(\alpha=2.2\), \(\mu=1.8\), and
    \(\lambda_k/\omega_k=0.25\).]{
        \includegraphics[width=0.465\textwidth]
        {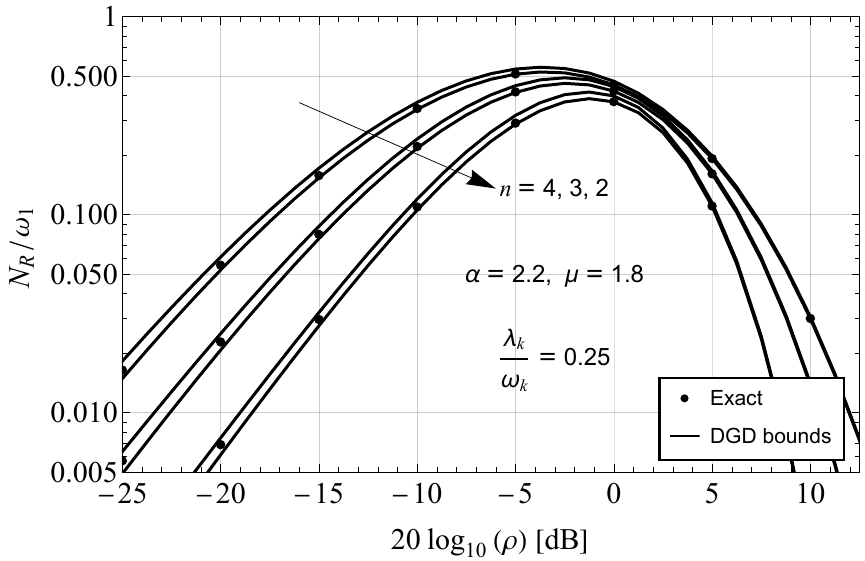}
        \label{fig:alpha_mu_DGD_bounds_weak_drift}
    }
    \hfill
    \subfloat[\(\alpha=1.5\), \(\mu=2.5\), and
    \(\lambda_k/\omega_k=0.5\).]{
        \includegraphics[width=0.465\textwidth]
        {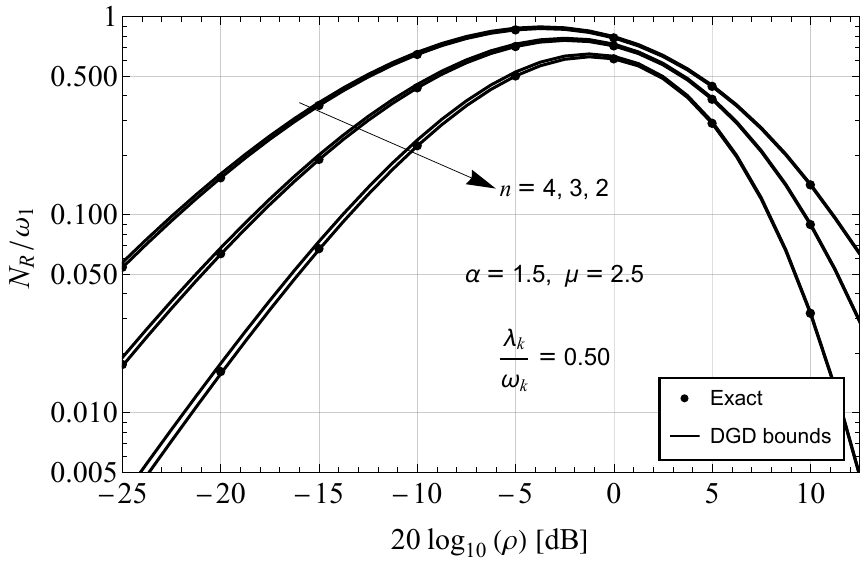}
        \label{fig:alpha_mu_DGD_bounds_strong_drift}
    }
    \caption{Exact LCR and DGD bounds for products of \(n=2,3,4\)
    quasi-identically distributed \(\alpha\)-\(\mu\) factors with equal
    drifts and Doppler shifts.}
    \label{fig:alpha_mu_DGD_bounds_n234}
\end{figure*}

For normalized \(\alpha\)-\(\mu\) fading \cite{yacoub2007alpha},
Table~\ref{tab:a-m_special_cases} lists representative special cases, and
\begin{align}
    f_0(r)
    &=
    \frac{\alpha\mu^\mu}{\Gamma(\mu)}
    r^{\alpha\mu-1}
    \euler^{-\mu r^\alpha},\\
    s_0(r)
    &=
    \frac{r^{1-\alpha/2}}{\alpha\sqrt\mu}.
\end{align}
Here \(\omega_k\) is the angular maximum Doppler shift. Define
\[
    \mathcal G_n(z;b_1,b_2)
    \triangleq
    \operatorname{G}_{0,n}^{n,0}
    \!\left(
    z\;\middle|\;
    \begin{matrix}
    -\\
    b_1,\underbrace{b_2,\ldots,b_2}_{n-1}
    \end{matrix}
    \right),
\]
where \(\operatorname{G}\) is the Meijer \(G\)-function
\cite[Eq.~(5.3.1)]{bateman1953higher}. Then, the closed forms needed for the DGD bounds in
\eqref{eq:quasi_DGD_linear_lower_bound} and
\eqref{eq:quasi_DGD_linear_upper_bound} 
are
\begin{align}
    f_R(u)
    &=
    \frac{\alpha\mu^n u^{\alpha-1}}{\Gamma^n(\mu)}
    \mathcal G_n(\mu^n u^\alpha;\mu-1,\mu-1),
    \label{eq:fR_a-m}\\
    \Sigma_0(u)
    &=
    \frac{1}{\Gamma^n(\mu)}
    \mathcal G_n\!\left(
        \mu^n u^\alpha;
        \mu-\frac12,
        \mu
    \right),
    \label{eq:Sigma_a-m}\\
    \Theta_0^2(u)
    &=
    f_R(u)
    \frac{
        \mu^{-n/\alpha}
    }{
        \alpha\Gamma^n(\mu)
    }
    \mathcal G_n\!\left(
        \mu^n u^\alpha;
        \mu+\frac{1 - \alpha}{\alpha},
        \mu+\frac{1}{\alpha}
    \right).
    \label{eq:Theta_a-m}
\end{align}
The derivations are detailed in Appendix~\ref{app:alpha_mu_quasi_DGD}.
The corresponding special cases follow by substituting \(\alpha\) and
\(\mu\) values of Table~\ref{tab:a-m_special_cases}.

For \(n=2\), these formulas simplify to
\begin{align}
    f_R(u)
    &=
    \frac{2\alpha\mu^{2\mu}}{\Gamma^2(\mu)}
    u^{\alpha\mu-1}\operatorname{K}_0(2\mu u^{\alpha/2}),\\
    \Sigma_0(u)
    &=
    \frac{\sqrt{\pi}\mu^{2\mu-1}}{\Gamma^2(\mu)}
    u^{\alpha\mu-\alpha/2}\euler^{-2\mu u^{\alpha/2}},\\
    \Theta_0^2(u)
    &=
    \frac{4\mu^{4\mu-1}}{\Gamma^4(\mu)}
    u^{2\alpha\mu-\alpha/2}
    \operatorname{K}_0(2\mu u^{\alpha/2})
    \operatorname{K}_1(2\mu u^{\alpha/2}).
\end{align}

Figure~\ref{fig:alpha_mu_DGD_bounds_n234} confirms the enclosure for
\(n=2,3,4\). With equal factor parameters,
\(|\lambda_\Sigma|/\|\boldsymbol\omega\|_2=\sqrt n\,
|\lambda_1|/\omega_1\). Indeed, the larger-drift case exhibits the tighter bounds
predicted by the asymptotic analysis.

\section{Conclusion}
\label{sec:conclusion}


This paper investigated the level crossing rate of products of smooth
nonstationary stochastic processes. First, universal bounds were derived for
positive Rice-regular factors, requiring only the individual LCRs and signed
level-crossing imbalances. These bounds provide model-independent guarantees
for product processes and clarify how each factor contributes to the total
crossing rate.

A sharper analysis was then developed for the class of DGD processes. For products of independent DGD factors,
an exact conditional representation of the product derivative was obtained,
leading to an exact integral expression for the product LCR. From this
representation, tractable lower and upper bounds were derived. The lower
bound depends only on the signed product imbalance
\(\Delta_{\mathcal R}\) and on the Euclidean norm of the averaged
derivative-spread vector, whereas the upper bound follows from a quadratic
folded-normal inequality. For the considered numerical examples, the DGD bounds provide a
substantially tighter enclosure than the universal bounds.
The DGD bounds can provide a substantially tighter enclosure, as confirmed by
the considered numerical examples.

The applications illustrated the usefulness of the proposed framework. In the nonstationary Double-Rayleigh model, deterministic scale variations induce a nonzero drift at the product level, even when the underlying Rayleigh envelopes are stationary. The DGD bounds closely bracket the exact LCR and become asymptotically tight in drift-dominated regimes. For quasi-identically distributed \(\alpha\)-\(\mu\) products with angular maximum Doppler shifts, the bounds reveal a simple decomposition: drift contributions combine coherently through the aggregate drift, while independent derivative spreads combine through the Euclidean Doppler norm. Closed forms were obtained for the quantities needed to evaluate these bounds, including several classical fading models as special cases.

Overall, the results provide both rigorous guarantees and practical
approximations for LCR analysis in nonstationary product-fading channels.
Possible extensions include ratio or sum of processes, correlated factors, and
specific non-Gaussian conditional derivative models.

\appendix

\subsection{Proof of Lemma~\ref{lem:Psi_upper_bound}}
\label{app:proof_upper_bound_Psi}

\begin{proof}
By evenness and positive homogeneity of $\Psi(a,b)$, it is enough to consider
\(b=1\) and \(a=x\ge0\). Let
\[
    m(x)\triangleq \Psi(x,1) = 2\phi(x)+x(2\Phi(x)-1).
\]
Then, its derivative $m'(x) = 2\Phi(x)-1$ and we can rewrite
\[
    m(x)=2\phi(x)+xm'(x).
\]
Define
\[
    g(x)\triangleq x^2+\frac{2}{\pi}-m^2(x).
\]
Since \(m(0)=\sqrt{2/\pi}\), we have \(g(0)=0\). Moreover,
\begin{align}
    g'(x)
    &=
    2(x-m(x)m'(x)) \nonumber\\
    &=
    2\left(
    x(1-(m'(x))^2)-2\phi(x)m'(x)
    \right)\nonumber\\
    &= 
    4\left(
    2x\Phi(x)\Phi^c(x)-\phi(x)m'(x)
    \right), \label{eq:g'}
\end{align}
where \(\Phi^c(x)\triangleq1-\Phi(x)\). Define the inverse Mills
ratio
\[
    v(x)\triangleq \frac{\phi(x)}{\Phi^c(x)} .
\]
Sampford's inequality \cite{sampford1953mills} gives \(0<v'(x)<1\) for every \(x\). Hence, for
\(x\ge0\),
\[
    v(x)-v(-x)
    =
    \int_{-x}^{x}v'(t)\,\d t
    \le
    2x .
\]
Since \(\phi(-x)=\phi(x)\) and \(\Phi^c(-x)=\Phi(x)\),
\[
    v(x)-v(-x)
    =
    \frac{\phi(x)}{\Phi^c(x)}
    -
    \frac{\phi(x)}{\Phi(x)}
    =
    \frac{\phi(x)m'(x)}
    {\Phi(x)\Phi^c(x)} .
\]
Therefore,
\[
    \phi(x)m'(x)
    \le
    2x \Phi(x)\Phi^c(x).
\]
Thus, from \eqref{eq:g'}, \(g'(x)\ge0\) for \(x\ge0\). Since \(g(0)=0\), it follows that
\(g(x)\ge0\), i.e.,
\[
    \Psi^2(x,1)
    \le
    x^2+\frac{2}{\pi},
\]
which completes the proof.
\end{proof}

\subsection{Proof of the closed form for $\overline{\N}_k$ in \eqref{eq:Nbar_closed}}
\label{app:Nk}

\begin{proof}
Let
$
    \overline{\N}_k(u)=I_1+I_2,
$
where, omitting $k$,
\begin{align}
    I_1
    &=
    2\rho\sqrt{\frac{2}{\pi}}\omega
    \int_0^\infty
    \exp\!\left[
    -x^2
    -
    \frac{\rho^2}{x^2}
    \left(
    1+\frac{2\lambda^2}{\omega^2}
    \right)
    \right]\d x,\\
    I_2
    &=
    4\lambda\rho^2
    \underbrace{\int_0^\infty
    \frac{1}{x}
    \erf\!\left(
    \frac{\sqrt{2}\lambda\rho}{\omega x}
    \right)
    \exp\!\left(
    -x^2-\frac{\rho^2}{x^2}
    \right)\d x}_{\triangleq J(\lambda)},
\end{align}
and $\erf(x)\triangleq 2\Phi(\sqrt{2}x)-1$. By \cite[Eq.~(3.478.4)]{GradshteynRyzhik2007},
\[
    I_1
    =
    \sqrt{2}\,\omega\rho
    \exp\!\left[
    -2\rho
    \sqrt{
    1+\frac{2\lambda^2}{\omega^2}
    }
    \right].
\]
For the second term, differentiating under the integral sign:
\[
    J'(\lambda)
    =
    \frac{
    \exp\!\left[
    -2\rho
    \sqrt{
    1+2\lambda^2/\omega^2
    }
    \right]
    }{
    \sqrt{(\omega^2/2)+\lambda^2}
    }.
\]
Since $J(0)=0$, using the change of variable
$
    \sqrt{2}\lambda=\omega\sinh y
$
yields
\[
    J(\lambda)
    =
    \int_0^{\operatorname{arsinh}(\sqrt{2}\lambda/\omega)}
    e^{-2\rho\cosh y}\,\d y .
\]
Therefore,
\[
    I_2
    =
    4|\lambda|\rho^2
    \operatorname{K}_0^{\rm L}\!\left(
    2\rho,
    \operatorname{arsinh}
    \frac{\sqrt{2}|\lambda|}{\omega}
    \right),
\]
which proves \eqref{eq:Nbar_closed}.
\end{proof}

\subsection{Proof of closed forms \eqref{eq:fR_a-m}, \eqref{eq:Sigma_a-m} and \eqref{eq:Theta_a-m}}
\label{app:alpha_mu_quasi_DGD}

\begin{proof}
Define
$
    T_i\triangleq \mu R_i^\alpha,~i=1,\ldots,n.
$
Then \(T_1,\ldots,T_n\) are i.i.d. gamma random variables with shape \(\mu\)
and unit rate, i.e.,
\[
    f_T(t)=\frac{1}{\Gamma(\mu)}t^{\mu-1}\euler^{-t},
    \qquad t>0.
\]
We shall also use
$
    z\triangleq \mu^n u^\alpha.
$
We use the standard Mellin-transform identity for products of independent
gamma random variables \cite{springer1970products}, that is,
\[
    f_{T_1\cdots T_n}(w)
    =
    \frac{1}{\Gamma^n(\mu)}
    G_{0,n}^{n,0}
    \!\left(
    w\;\middle|\;
    \begin{matrix}
    -\\
    \underbrace{\mu-1,\ldots,\mu-1}_{n}
    \end{matrix}
    \right).
\]
Since
\(\prod_i T_i=\mu^n R^\alpha\), then
\[
    f_R(u)
    =
    \alpha\mu^n u^{\alpha-1}
    f_{\prod_i T_i}(\mu^n u^\alpha),
\]
which gives
\[
    f_R(u)
    =
    \frac{\alpha\mu^n u^{\alpha-1}}{\Gamma^n(\mu)}
    \mathcal G_n(\mu^n u^\alpha;\mu-1,\mu-1).
\]

We now derive \(\Sigma_0(u)\). Let
$
    W_1\triangleq \prod_{i=2}^n T_i.
$
Then,
\[
    s_0\!\left(\frac{u}{C_1}\right)
    f_0\!\left(\frac{u}{C_1}\right)
    =
    \frac{
        z^{\mu-1/2}
    }{
        \Gamma(\mu)
    }
    W_1^{-(\mu-1/2)}
    \euler^{-z/W_1}.
\]
Taking expectation over \(T_2,\ldots,T_n\) gives
\begin{align*}
    \Sigma_0(u)
    &=
    \frac{z^{\mu-1/2}}{\Gamma^n(\mu)}
    \int_{\mathbb R_+^{n-1}}
    \exp\!\left(
        -\sum_{i=2}^n t_i
        -\frac{z}{\prod_{i=2}^n t_i}
    \right)
    \prod_{i=2}^n t_i^{-1/2}\,\d \mathbf t .
\end{align*}
By the integral representation of the Meijer G-function \cite[Section~5.3]{bateman1953higher}, this equals
\[
    \Sigma_0(u)
    =
    \frac{1}{\Gamma^n(\mu)}
    G_{0,n}^{n,0}
    \!\left(
        z
        \;\middle|\;
        \begin{matrix}
        -\\
        \mu-\frac12,\underbrace{\mu,\ldots,\mu}_{n-1}
        \end{matrix}
    \right).
\]

It remains to derive \(\Theta_0^2(u)\). 
Since
\[
    C_1
    s_0^2\!\left(\frac{u}{C_1}\right)
    f_0\!\left(\frac{u}{C_1}\right)
    =
    \frac{
        \mu^{-n/\alpha}
    }{
        \alpha\Gamma(\mu)
    }
    z^{\mu+1/\alpha-1}
    W_1^{1-\mu}
    \euler^{-z/W_1}
\]
then
\begin{align*}
    &\E\!\left[
        C_1s_0^2\!\left(\frac{u}{C_1}\right)
        f_0\!\left(\frac{u}{C_1}\right)
    \right]\\
    &\quad =
    \frac{\mu^{-n/\alpha}z^{\mu+1/\alpha-1}}
    {\alpha\Gamma^n(\mu)}
    \int_{\mathbb R_+^{n-1}}
    \exp\!\left(
        -\sum_{i=2}^n t_i
        -\frac{z}{\prod_{i=2}^n t_i}
    \right) \d \mathbf t.
\end{align*}
Again using the Meijer G-function integral representation,
\begin{align*}
    &\E\!\left[
        C_1s_0^2\!\left(\frac{u}{C_1}\right)
        f_0\!\left(\frac{u}{C_1}\right)
    \right]
    =
    \frac{\mu^{-n/\alpha}}{\alpha\Gamma^n(\mu)}
    \mathcal G_n\!\left(
        z;
        \mu+\frac{1}{\alpha}-1,
        \mu+\frac{1}{\alpha}
    \right).
\end{align*}
Multiplying by \(f_R(u)\) 
completes the proof.
\end{proof}





\bibliographystyle{IEEEtran}
\bibliography{ref} 

\end{document}